\documentclass[a4paper,10pt]{article}
\usepackage{graphics} 
\usepackage{graphicx,color}
\usepackage[active]{srcltx}
\usepackage{amsmath} 
\usepackage{amssymb}  
\usepackage{amsmath,amssymb,amsfonts,theorem}
\usepackage{enumerate}
\usepackage{tikz}
\usepackage{epstopdf}

{ \theorembodyfont{\normalfont} 

\newcommand{\R}{\mathbb{R}}

\newcommand{\diag}{\operatorname{diag}} 

\newcommand\ie{\emph{i.e.}}

{

\newtheorem{remark}{Remark}
}
\newtheorem{assumption}{Assumption}
\newtheorem{definition}{Definition}

\newtheorem{lemma}{Lemma}
\newtheorem{corollary}{Corollary}
\newtheorem{proposition}{Proposition}

\def\be{\begin{equation}}
\def\ee{\end{equation}}
\def\ba{\begin{array}}
\def\ea{\end{array}}
\def\eqa{\begin{eqnarray}}
\def\eqe{\end{eqnarray}}
\title{The interconnection of quadratic droop voltage controllers is a Lotka-Volterra system: implications for stability analysis}
\author{Matin Jafarian,
 Henrik Sandberg,
 Karl H. Johansson
\thanks{This work was supported in part by the Knut and Alice Wallenberg Foundation, the Swedish Strategic Research Foundation, the Swedish Research Council, and the Swedish Energy Agency. M. Jafarian, H. Sandberg, K.H. Johansson are with the Automatic Control Department, School of Electrical Engineering, KTH Royal Institute of Technology, Stockholm, Sweden. Email: {\tt\small {matinj@kth.se}, {hsan@kth.se}, {kallej@kth.se}.}}}
\date{}
\begin{document}
\maketitle

\thispagestyle{empty}
\pagestyle{empty}
\begin{abstract}
This paper studies the stability of voltage dynamics for a power network in which nodal voltages are controlled by means of quadratic droop controllers with nonlinear AC reactive power as inputs. We show that the voltage dynamics is a Lotka-Volterra system, which is a class of nonlinear positive systems. We study the stability of the closed-loop system by proving a uniform ultimate boundedness result and investigating conditions under which the network is cooperative. We then restrict to study the stability of voltage dynamics under a decoupling assumption (\ie,\ zero relative angles). We analyze the existence and uniqueness of the equilibrium in the interior of the positive orthant for the system and prove an asymptotic stability result.
\end{abstract}
\section{Introduction}\label{sec:int}
The recent interest in integrating distributed generation in power systems has motivated the design of new control techniques for assuring desired performance, for instance, maintaining appropriate voltage levels. Voltage control in various problem settings have been widely studied in the literature, e.g., \cite{andreasson2017performance,de2017bregman,matincdc2016,
schiffer2014conditions,vasquez2009adaptive} to name a few.
In general, the physical model of electrical power systems can be described using four main variables: active power, reactive power, voltage magnitude and angle. The way these variables are interacting in an AC power network is defined by the (nonlinear) {\em AC power flow} model \cite{kundur1994power}. It follows from this model that voltages and angles depend on both active and reactive power flows. However, most designs for controlling voltage (angle) dynamics rely on a {\em decoupling} assumption where voltage (angle) depends only on the reactive (active) power. A decoupled, local and linearized AC power flow model for lossless power networks is the so-called {\em DC power flow} model which is the assumption behind the design of conventional droop controllers. Recently, a quadratic droop controller was introduced in \cite{jw16} in order to include the quadratic nature of the reactive power flow in a decoupled power flow model for an inductive network. Although the assumption behind designing (quadratic) droop controllers is not the original AC power flow model, studying the use of such controllers with this power flow model, which includes the power losses and does not restrict the size of relative angles, is interesting from both theoretical and practical point of views. A linearized model of a network of quadratic droop controllers whose injected reactive power obeys the AC power flow model was considered in \cite{teixeira2015voltage} where it is shown that the linearized time-invariant system is a stable positive system provided some constraints on the relative angles, controller gain and the power line parameters hold. Positive systems are a class of dynamical systems whose state remain non-negative, if their initial condition is non-negative. The fact that the sign of the voltage magnitude is positive motives studying the voltage dynamics from a positive system perspective.\\[1mm] 
{\bf Main contributions:} This paper considers a power network in which nodal voltages are controlled by means of the quadratic droop controllers and studies the stability within the framework of positive systems. First, we show that interconnected quadratic droop controllers with nonlinear injected reactive power can be represented as a Lotka-Volterra system, which is traditionally studied in mathematical biology.
Second, we investigate the dynamical properties of the network with time-varying voltage angles, droop gains, and references. We prove boundedness of the solutions. Third, we consider the special case where a decoupling assumption holds (\ie,\ zero relative angles) and study the conditions under which the system possesses a unique equilibrium in the interior of the positive orthant. We also provide a Lyapunov-based argument to prove asymptotic stability of the equilibrium.

Compared to previous works (e.g., \cite{de2017bregman,jw16,teixeira2015voltage}), our contribution is to shed a new light on inherent dynamical properties of a network of quadratic droop controllers. Moreover, we analyze the stability of the network from a nonlinear positive system point of view which requires the application of completely different analytical tools.

The paper is organized as follows. Section \ref{sec:pre} presents preliminaries and problem formulation. Section \ref{sec:model} reveals the structure of the nonlinear positive system. Boundedness of the time-varying lossy network and its cooperative property is discussed in Section \ref{sec:ana}. Stability of the network under the decoupling assumption is analyzed in Section \ref{sec:dec}. Section \ref{sec:sim} presents simulation results and Section \ref{sec:con} concludes the paper.\\
{\bf{Notation}}\\
Let $\R_{+}=[0,+\infty)$ and $\R_{+}^{0}=(0,+\infty)$, while $\R^n_{+}$ and ${\rm int}(\R^n_{+})$ are the set of $n$-tuples for which all components belong to $\R_{+}$ and $\R_{+}^{0}$, respectively. The boundary of $\R^{n}_{+}$ is denoted by ${\rm bd}(\R^n_{+})$. The notation $\diag (x)$ is the $n \times n$ diagonal matrix whose entries are the elements of $x \in \R^n$.
\section{Preliminaries and problem formulation}
\subsection{Preliminaries}\label{sec:pre}
Consider the following differential equations 
\be\label{eq:p1}
\dot x(t)= f(x(t)),
\ee
\be\label{eq:p11}
\dot x(t)= F(x(t),t),
\ee
with $x \in R^n$, $f: \R^n \rightarrow \R^n$ , $F:\R^n \times \R \rightarrow \R^n$. The solution of \eqref{eq:p1} or \eqref{eq:p11} at time $t$ with initial condition $(x_0,t_0)$ is denoted by $x(t,t_0,x_0)$ where the equation will be clear from the context. The following definitions are used throughout the paper \cite{angeli2003monotone,de2001stability,peuteman2000boundedness}.
\begin{definition}[Positive systems]\label{def:po}
System \eqref{eq:p1}, \eqref{eq:p11} is positive iff $\R^n_{+}$ is forward invariant.
\end{definition}
\begin{lemma}\label{lem:po}
The following property is a necessary and sufficient condition for positivity of system \eqref{eq:p1},\\ 
\be\label{eq:prop}
\forall x \in {\rm bd}(\R^n_{+}): x_i=0 \Rightarrow f_i(x)\geq 0.
\ee
\end{lemma}
\begin{definition}\label{def:met}
A matrix $A_{n \times n}$ is Metzler if its off-diagonal entries $a_{i,j} , \forall i \neq j$ are non-negative. Similarly,
$A(t)$ is Metzler if $a_{i,j}(t) , \forall i \neq j$ are non-negative. 
\end{definition}
\begin{definition}\label{def:co}
The map $f(x)$ in \eqref{eq:p1} is cooperative in $\R^n_{+}$ if the Jacobian matrix $\frac{\partial f}{\partial x}$ is Metzler for all $x \in \R^n_{+}$. A similar definition holds for System \eqref{eq:p11} (see Definition 2.2. in \cite{cui2001permanence}).
\end{definition}
\begin{definition}\label{def:hom} Given $r=(r_1,\ldots,r_n), \forall i, r_i >0$, define the dilation map $\delta: \R_{+} \times \R^{n} \rightarrow \R^{n}$ as follows
\be
\delta: (s,x) \rightarrow \delta(s,x)=(s^{r_1} x_1, \ldots, s^{r_n} x_n),
\ee
where $x=(x_1,\ldots,x_n)$. A continuous function $F: \R^n \times \R \rightarrow \R^n$ is $r$-homogeneous of order $\tau \geq 0$ if
\be
\forall x \in \R^n, \forall t \in \R, \forall s \in \R_{+}: F(\delta(s,x),t)= s^{\tau} \delta(s,F(x,t)).
\ee
\end{definition}
\begin{definition}[Uniform boundedness]
System \eqref{eq:p11} is uniformly bounded if $\forall R_1>0$, there exists an $R_2(R_1) > 0$ such that $\forall x_0 \in \R^n, \forall t_0, \forall t \geq t_0$
$$||x_0|| \leq R_1 \ \ \Rightarrow ||x(t,t_0,x_0)|| \leq R_2(R_1).$$
\end{definition}
\begin{definition}[Uniform ultimate boundedness]
System \eqref{eq:p11} is uniformly ultimately bounded if there exists an $R>0$ such that $\forall R_1>0$, there exists a $T(R_1) > 0$ such that $\forall x_0 \in \R^n, \forall t_0,\forall t \geq t_0+T(R_1)$
$$||x_0|| \leq R_1 \ \ \Rightarrow ||x(t,t_0,x_0)|| \leq R.$$
\end{definition}
\begin{definition}[r-homogeneous norm]\label{def:norm}
The r-homogeneous norm $\rho: \R^n \rightarrow \R$ is given by
$$\rho(x)= \sum_{i=1}^{n} |x_i|^{\frac{1}{r_i}}$$
where $0 < r_i <1$.
\end{definition}
\subsection{Problem formulation}\label{sec:pf}
Consider a power network composed of $n$ busbars and $m$ power lines. Let the network be modeled as a connected, undirected graph with $n$ nodes and $m$ edges. The nodal reactive power obeys the AC power flow model \cite{kundur1994power},\ \ie
\begin{equation}\label{eq:ac}
Q_i= -B_{i} V_i^2 + \sum_{j \in {\cal N}_i} (B_{i,j} V_i V_j \cos(\theta_{i,j})-G_{i,j} V_i V_j \sin(\theta_{i,j}),
\end{equation}
where $Q_i$, $V_i$ and $\theta_{i}$ are the reactive power, voltage magnitude and voltage angle of busbar $i$, respectively. Also, ${\cal N}_i$ denotes the set of neighbors of node $i$. The variable $\theta_{i,j}$ is the relative angle, \ie, $\theta_{i,j} := \theta_i-\theta_j$. Variables $G_{i,j} \geq 0$, $B_{i,j} \leq 0$ are the conductance and susceptance of the line $(i,j)$, which connects busbar $i$ to busbar $j$, $G_{i,j} = G_{j,i}$ and $B_{i,j} = B_{j,i}$. Furthermore, $B_{i}= B_i^{sh}+\sum_{j \in {\cal N}_i} B_{i,j}$ where $B_i^{sh}$ denotes the shunt susceptance. Notice that $G_{i,j} \geq 0$, $B_i^{sh} \geq 0$ and $B_{i,j} \leq 0$. It is a common assumption to consider $B_i^{sh} \ll \sum_{j \in {\cal N}_i} |B_{i,j}|$, hence $B_i \leq 0$.
We assume that each node of the network is connected to an inverter, which is modeled as a controllable voltage source \cite{jw16}. We assume that nodal voltages are controlled by means of quadratic droop voltage controllers, designed to incorporate the quadratic nature of reactive power in a conventional droop controller 
as follows
\begin{equation}\label{eq:qdr5}
\tau_i \dot {V_i}= V_i (-k_i (V_i - V_i^\ast))- u_i,
\end{equation}
where $\tau_i > 0, k_i >0, u_i \in \R$, and $V_i^\ast>0 $ are the controller's time constant, droop gain, input, and the nominal voltage of node $i$, respectively. In \cite{jw16}, the control input, $u_i$, is designed to be equal to the nodal reactive power of a simplified power flow model obtained from \eqref{eq:ac} by imposing the decoupling assumption $\theta_{i,j}=0$,
 \ie,
\begin{equation}\label{eq:qdr2}
\tau_i \dot {V_i}= V_i (-k_i (V_i - V_i^\ast))+B_{i} V_i^2 - \sum_{j \in {\cal N}_i} B_{i,j} V_i V_j.
\end{equation}
In this paper, we consider the controller in \eqref{eq:qdr5} and replace $u_i$ with the general AC reactive power flow as in \eqref{eq:ac}. Thus,
\begin{equation}\label{eq:qdr}
\tau_i \dot {V_i}= V_i (-k_i (V_i - V_i^\ast))- Q_i.
\end{equation} 
This paper first considers the controller \eqref{eq:qdr} and study its dynamical properties from a positive system point of view. Second, we study the conditions under which there exists a stable equilibrium in ${\rm int}(\R^n_{+})$ for the network with nodal controllers as in \eqref{eq:qdr2} within the framework of positive systems.
\section{Voltage dynamics as a Lotka-Volterra system}\label{sec:model}
Lotka-Volterra systems are a class of nonlinear positive systems with the dynamics
\begin{equation}\label{eq:kolb}
\dot x= \diag (x) (f(x)+b).
\end{equation}  
where $x \in \R^n$ and $b \in {\rm int}(\R^n_{+})$ \cite{de2001stability}. Now, let us consider a power network with 
each node connected to a quadratic droop controller as introduced in the previous section. We consider the controller \eqref{eq:qdr} which is a generalization compared with \eqref{eq:qdr2} due to the injection of reactive power flow in \eqref{eq:ac}. By replacing $Q_i$ from \eqref{eq:ac} in \eqref{eq:qdr5}, the voltage dynamics of each node is 
\be\begin{aligned}\label{eq:two2}
\tau_i \dot {V_i}= V_i \Bigg[&-k_i (V_i - V_i^\ast)- |B_i| V_i + \sum_{j \in {\cal N}_i} V_j (G_{i,j} \sin \theta_{i,j}+ |B_{i,j}| \cos \theta_{i,j})\Bigg].
\end{aligned}\ee
Notice that $- B_{i,j}$ and $B_i$ in \eqref{eq:ac} are replaced by $|B_{i,j}|$ and $-|B_i|$ in \eqref{eq:two2} since $B_{i,j} \leq 0$ and $B_i < 0$. Now, let us rewrite \eqref{eq:two2} in the form of \eqref{eq:kolb}. We have
\be\begin{aligned}\label{eq:two}
\tau_i \dot {V_i}= V_i \Bigg[&\sum_{j \in {\cal N}_i} V_j (G_{i,j} \sin \theta_{i,j}+ |B_{i,j}| \cos \theta_{i,j})
- (k_i+|B_i|) V_i + k_i V_i^\ast\Bigg].
\end{aligned}\ee
Denote $\sin\theta_{i,j}$, $\cos\theta_{i,j}$ by $\Delta_{i,j}^s$, $\Delta_{i,j}^c$, respectively. Thus, 
$\Delta_{i,j}^s= -\Delta_{j,i}^s$, $\Delta_{i,j}^c = \Delta_{j,i}^c$ and 
$$\Delta_{i,j}^s \in [-1,1],\quad \Delta_{i,j}^c \in [-1,1].$$ Writing the equation in \eqref{eq:two} for all nodes, we obtain
\begin{equation}\begin{aligned}\label{eq:net}
\diag(\tau) \begin{bmatrix} \dot{V_1}\\ \dot{V_2}\\ \vdots \\ \dot{V_n} \end{bmatrix}=
\diag(V)
\Bigg(\begin{bmatrix} f_1(V,\theta)\\ f_2(V,\theta) \\ \vdots \\ f_n(V,\theta) \end{bmatrix}+
\begin{bmatrix} b_1 \\ b_2 \\ \vdots \\ b_n \end{bmatrix}\Bigg), 
\end{aligned}\end{equation}
where $\tau= (\tau_1,\tau_2,\ldots,\tau_n)^T$, $V= (V_1,V_2,\ldots,V_n)^T$, $\dot V=({\dot V}_1,{\dot V}_2,\ldots,{\dot V}_n)^T$,
$b_i=k_i V_i^\ast$, and
$$f_i(V,\theta)= -(|B_i|+k_i) V_i + \sum _{j \in {\cal N}_i} V_{j} (G_{i,j} \Delta_{i,j}^s + |B_{i,j}| \Delta_{i,j}^c).$$
Let us rewrite $f(V,\theta)$ as $f(V,\theta)= \Psi(\theta(t)) V$
where $\Psi(\theta(t))$ is the following matrix
\begin{equation}\label{eq:net2}
{\small{\begin{bmatrix}-(|B_1|+k_1)&\ldots& G_{1,n} \Delta_{1,n}^s+ |B_{1,n}| \Delta_{1,n}^c)
\\ \vdots&\vdots&\vdots
\\ -G_{1,n} \Delta_{1,n}^s+ |B_{1,n}| \Delta_{1,n}^c&\ldots & -(|B_n|+k_n)\end{bmatrix}}}.
\end{equation}
In compact form, the network model is
\begin{equation}\label{eq:netc}
\diag(\tau) \dot V= \diag(V)(\Psi(\theta(t))\;V+b),
\end{equation}
with $b=(k_1 V_1^\ast, \ldots, k_n V_n^\ast)^T$. Matrix $\Psi$ is called the {\em interaction} matrix \cite{zhao2010classification}. 
\begin{proposition}\label{pr1}
System \eqref{eq:netc} is positive. That is, $\forall V(0) \in \R^n_{+}$ and $\forall \theta_{i,j} \in \R$, $V(t) \in \R^n_{+}$.
\end{proposition}
\begin{proof}
The proof is based on the Definition \ref{def:po}. Consider $V(0) \geq 0$. If there exists $V_i(0)=0$, it is immediate to see that ${\dot V}_i=0$. If $V_i(0)>0$, as the system evolves, ${\dot V}_i$ could be zero, positive or negative. If ${\dot V}_i >0$, $V_i$ grows in $\R^n_{+}$. If ${\dot V}_i=0$, $V_i$ stays in $\R^n_{+}$. If ${\dot V}_i<0$, $V_i$ decreases. Due to the continuity of ${\dot V}_i$ in \eqref{eq:net}, the decrease lead to $V_i=0$, thus $V_i$ cannot decrease further. Hence, $\R^n_{+}$ is forward invariant for \eqref{eq:net} which ends the proof.
\end{proof}
\begin{remark} 
The above is a general result compared with \cite{teixeira2015voltage} which has shown the positivity of the linearized system assuming ${\dot\theta}_{i,j}=0$ and imposing constraints on $\frac{G_{i,j}}{B_{i,j}}$ ratio.
\end{remark}
{\em Properties of Lotka-Volterra systems}\\
A Lotka-Volterra system with interaction matrix $\Psi$ is \cite{zhao2010classification}
\begin{itemize}
\item {\em cooperative (competitive)} if $\Psi_{i,j} \geq 0$ ($\Psi_{i,j} \leq 0$) for all $i \neq j$, (similar to Definition \ref{def:co}),
\item{\em dissipative} if there exists a diagonal matrix $D >0$ such that, $\Psi D \leq 0$, and {\em stably dissipative} if it stays dissipative under small enough perturbation $\delta_i>0$ of its non-zero elements.
\end{itemize}
In cooperative networks, in contrast to competitive networks, agents (nodes) benefit from interacting with each other. Properties of a cooperative system allow us to derive conditions for existence of a unique equilibrium in $\rm{int}(\R^n_{+})$.
Also, inspired by results of competition of ecological species, we envision that voltage drop
could be studied under the competitive system assumption. The latter is under our current investigations and requires further analysis. Dissipativity is useful in studying the convergence behavior for a large scale network specially when the network is heterogeneous. Although the analysis of this paper do not directly rely on this property, in the Section \ref{sec:dec}, we discuss that the network under a decoupling assumption is stably dissipative for the sake of comprehensiveness and future extensions.
\section{Analysis: The case of lossy network}\label{sec:ana}
This Section considers the system in \eqref{eq:netc} with the interaction matrix $\Psi$ in \eqref{eq:net2}. This section assume a lossy network with controller in \eqref{eq:qdr}, \ie\ ${\dot\theta}_{i,j} \neq 0$ and $G_{i,j} \neq 0$. We first assume that ${\dot V}_i^\ast \neq 0$, $\dot k_i \neq 0$, \ie,
\begin{equation}\label{eq:netcT}
\diag(\tau) \dot V= \diag(V) \Psi(\theta(t)) V + \diag(k(t)) V^\ast(t)),
\end{equation}
where $V^\ast(t)=(V^\ast_1(t),\ldots,V^\ast_n(t))^T$. Our aim is to study the boundedness of voltage trajectories in a control-theory sense. We differentiate ultimate boundedness in a control-theory sense from the voltage stability in a power-system sense. The former implies that voltage magnitudes are bounded and ultimately converge to a ball in $\R^n_{+}$ with radius $R$, while the latter requires steady desired bounds \cite{kundur1994power}. This paper studies the boundedness of the closed-loop system without determining the bounds. We also show the usage of tools from the positive systems framework in the analysis of power systems which is interesting from a theoretical point of view. We first allow no restriction on $\theta_{i,j}$ and establish a uniform boundedness result for voltage trajectories. Notice that although variations of $\theta_{i,j}$ depend on voltage magnitudes based on the physical laws, the results of this section are independent of these effects. In fact, the variations of the relative angles will cause variations in $\Delta_{1,2}^{c}$ and $\Delta_{1,2}^{s}$, which are both bounded and take a value in the set $[-1,+1]$, in $\Psi(\theta(t))$ \eqref{eq:netcT}. Thus, without making any specific assumption on the dynamics of $\theta_{i,j}$, we can mathematically model the variations of $\theta_{i,j}$ as a time varying variable which takes a value in $[-1,+1]$.\\[2mm]   
Consider system \eqref{eq:netcT} with the general form 
$$\dot x= f(x(t),t)+g(x(t),t).$$ To study the boundedness of the system, we adopt the approach of \cite{peuteman2000boundedness} allowing us to study the time-invariant `frozen' system $\dot x= f(x(t),\sigma)+g(x(t),\sigma)$, \ie
\begin{equation}\label{eq:netcTf}
\diag(\tau) \dot V= \diag(V) \Psi(\theta(\sigma)) V + \diag(k(\sigma)) V^\ast(\sigma)),
\end{equation}
where $\sigma \in \R$ is treated as a constant parameter. The approach in \cite{peuteman2000boundedness} discusses the stability of homogeneous time-varying systems of a positive order (see Definition \ref{def:hom}) as well as a class of non-homogeneous time-varying systems which possesses a homogeneous approximation when the system state (e.g., $||V||$) is sufficiently large, \ie,\ system \eqref{eq:netcT}. First let us write $\Psi(\theta(\sigma))$ in \eqref{eq:net2} as $\Psi=\Psi^s+\Psi^c$, hence,
\begin{equation}\begin{aligned}\label{eq:net3}
\Psi= \begin{bmatrix} -(|B_1|+k_1)& |B_{1,2}| \Delta_{1,2}^{c,\sigma} & \ldots & |B_{1,n}| \Delta_{1,n}^{c,\sigma} \\ |B_{1,2}|\Delta_{1,2}^{c,\sigma} &-(|B_2|+k_2) & \ldots & |B_{2,n}|\Delta_{2,n}^{c,\sigma} \\ \vdots & \vdots & \cdots & \vdots \\|B_{1,n}| \Delta_{1,n}^{c,\sigma} & |B_{2,n}|\Delta_{2,n}^{c,\sigma} & \ldots & -(|B_n|+k_n)\end{bmatrix}
+\begin{bmatrix} 0 & G_{1,2} \Delta_{1,2}^{s,\sigma} & \ldots & G_{1,n} \Delta_{1,n}^{s,\sigma} \\ -G_{1,2} \Delta_{1,2}^{s,\sigma} & 0 & \ldots & G_{2,n} \Delta_{2,n}^{s,\sigma} \\ \vdots & \vdots & \cdots & \vdots \\ -G_{1,n} \Delta_{1,n}^{s,\sigma} & -G_{2,n} \Delta_{2,n}^{s,\sigma} & \ldots & 0\end{bmatrix},
\end{aligned}\end{equation}
where $\Delta_{i,j}^{c,\sigma}$ is the value of $\Delta_{i,j}^{c}$ at $t=\sigma$ and  $\Delta_{i,j}^{c,\sigma} \in [-1,+1]$ (a similar definition holds for $\Delta_{i,j}^{s,\sigma}$).\\[2mm]
We now prove the asymptotic stability of $\dot V=diag(V) \Psi(\theta(\sigma)) V$. This result is required in the proof of boundedness of the time-varying network \eqref{eq:netcT}.
\begin{proposition}\label{pr2}
If $\forall i: k_i >0$, then $\forall x \in \R^n, x \neq 0$, it holds that $x^T \Psi(\theta(\sigma)) x <0$.
\end{proposition}
\begin{proof}
Consider \eqref{eq:net3}. Observe that $\Psi^s$ is skew-symmetric. If $\Psi^c$ is negative definite, then $\Psi$ is Hurwitz and $x^T \Psi(\theta(\sigma)) x <0$. Applying the Gershgorin Circle Theorem \cite{varga2010gervsgorin}, a sufficient condition for $\Psi^c$ to be negative definite is that 
$$\forall i \in \{1,\ldots,n\}: |B_i|+k_i > \sum_{j \in {\cal N}_i} |B_{i,j} \Delta_{i,j}^{c,\sigma}|.$$
Recall that $|B_i|= B_i^{sh}+ \sum_{j \in {\cal N}_i} |B_{i,j}|$ and $\Delta_{i,j}^{c,\sigma} \in [-1,+1]$. Hence, the above is satisfied if $k_i >0$. 
\end{proof}
\begin{proposition}\label{pr3}
System $\dot V=diag(V) \Psi(\theta(\sigma)) V$ is positive and asymptotically stable at the origin.
\end{proposition}
\begin{proof}
From Lemma \ref{lem:po}, it is immediate to see that system $\dot V=diag(V) \Psi(\theta(\sigma)) V$ is positive. Take ${\cal V}= \sum_{i} |V_i|$ (where $|.|$ is the absolute value) as the Lyapunov candidate. Since ${\cal V}$ is not differentiable at the origin, we use tools from the nonsmooth theory, \ie\ the Clarke generalized gradient and set-valued derivative in order to calculate $\dot {\cal V}$ (see for example \cite{bacciotti1999stability}). Define the Clarke generalized gradient as follows
\be\partial {\cal V}= \{ p^V \quad{\rm s.t.}\quad  p^V_i \in \left\{\ba{lll} +1 &if&  V_i>0,\\ \left[-1,+1\right] &if& V_i=0 \ea\right.\}.\ee
The set-valued derivative is then obtained from ${\cal {\dot{\bar V}}}= \{a \in \R: a=\langle\dot V, p^V\rangle, \forall p^V \in \partial{\cal V}\}$ where $\langle,\rangle$ is the inner product. 
Since for $V_i=0$, it holds that ${\dot V}_i=0$, we obtain ${\cal {\dot{\bar V}}}=\{V^T \Psi(\theta(\sigma)) V\}$. Based on Proposition \eqref{pr2}, ${\cal {\dot{\bar V}}} \subseteq (-\infty,0]$. Applying (nonsmooth) La Salle's invariance principle \cite{bacciotti1999stability,jafarian2015formation}, the system is asymptotically stable at the origin.
\end{proof}
Now, we continue with proving uniform ultimate boundedness of system \eqref{eq:netcT}. 
\begin{assumption}\label{assT}
For system \eqref{eq:netcT},\\
1- there exists $c_k > 0$ such that for all $\sigma \in \R$ and for all $i$, $0<k_i (\sigma)<c_k$ holds, (boundedness of droop gains)\\
2- there exists $c_r > 0$ such that for all $\sigma \in \R$ and for all $i$, $|k_i (\sigma) V^\ast_i(\sigma)|<c_r$ holds (boundedness of references).
\end{assumption}
\begin{proposition}\label{pr4}
If Assumption \ref{assT} holds, then the time-varying system \eqref{eq:netcT} is uniformly and uniformly ultimately bounded.
\end{proposition}
\begin{proof}
The proof is based on Theorem 4.1 of \cite{peuteman2000boundedness}, which is an extension of Theorem 3.2, of \cite{peuteman2000boundedness}. Based on Theorem 4.1 \cite{peuteman2000boundedness}, the following conditions should hold for $f_H(V,t)= diag (V)(\Psi^s(t) + \Psi^c(t)) V$,\\
\begin{itemize}
\item $f_H(V,t)$ is homogeneous of order $\tau >0$: based on the Definition \ref{def:hom}, let us take $\delta^r_{\lambda}(V)=(\lambda^{r} V_1,\ldots,\lambda^{r} V_n)^T$, then $f_H(V,t)$ is r-homogeneous of order $\tau=r>0$,
\item $f_H(V,\sigma)$ is continuously differentiable with respect to $V$ and $\sigma$: this clearly holds,
\item there exists a $c_f > 0$ such that for all $\sigma \in \R$, for all $y \in \R^n$ with $\rho(y)=1$ (see Definition \ref{def:norm}), and $\forall i,k$, the following hold\\ $|f_H^i(y,\sigma)| \leq c_f$, $|\frac{\partial f_H^i}{\partial x_k}(y,\sigma)| \leq c_f$, $|\frac{\partial f_H^i}{\partial \sigma}(y,\sigma)| \leq c_f$. Considering Assumption \ref{assT}, the above conditions are satisfied since all elements of $\Psi^s(\sigma)$ and $\Psi^c(\sigma)$ are bounded,
\item each frozen system $\dot V=f_H(V,\sigma)$ is asymptotically stable at the origin:
this holds based on Proposition \ref{pr3},
\item there exists an $R_g >0$ and a continuous nonincreasing function $F:\R_{+} \rightarrow \R$ with $\lim _{s \rightarrow \infty} F(s)=0$ such that for all $V \in \R^n$ with $\rho(V) > R_g$ and $\forall t \in \R$, $$||\delta^r_{\rho(V)^{-1}}(\diag(V) \diag(k(t)) V^\ast(t))|| \leq \rho(V)^\tau F(\rho(V)).$$
To fulfill the above, that is the condition 4.1 of \cite{peuteman2000boundedness}, take $F(s)= \frac{\sqrt{n} c_r}{s^r}$ \cite{peuteman2000boundedness}, where $c_r$ is the upper bound of $k_i(t) V^\ast_i(t)$ by Assumption \ref{assT}. Based on the definitions of $\delta$ and $\rho$ (see Preliminaries), this last condition is also satisfied which ends the proof.
\end{itemize}
\end{proof}
Now, consider the system in \eqref{eq:netcT} assuming ${\dot V}_i^\ast=0$, $\dot k_i=0$ which gives the system in \eqref{eq:netc}. We conclude the ultimate boundedness of \eqref{eq:netc} based on the above proposition.
\begin{corollary}\label{cor1}
If $k_i >0, V_i^\ast>0, b_i <c_f$, then the system \eqref{eq:netc} is uniformly and uniformly ultimately bounded.
\end{corollary} 
Next, we assume boundedness of $\theta_{i,j}$ and verify the conditions under which system \eqref{eq:netc} is cooperative. This property allows us to derive conditions under which all voltage trajectories will converge to a ball in the interior of the positive orthant \ie \ away from zero. 
\begin{assumption}\label{ass3}
The relative voltage angles are bounded, e.g. $\theta_{i,j} \in [-\beta, \beta]$ for some constant $\beta$.
\end{assumption}
\begin{proposition}\label{pr5}
If Assumption \ref{ass3} holds and $\forall i,j: |\frac{G_{i,j}}{B_{i,j}}| < |\cot(\theta_{i,j})|$, then system \eqref{eq:netc} is cooperative.
\end{proposition}
\begin{proof} 
Based on Definition \ref{def:co} (and Definition 2.2. in \cite{cui2001permanence}), system \eqref{eq:netc} is cooperative if the interaction matrix $\Psi$ is Metzler (see Definition \ref{def:met}). To satisfy this condition, both $|B_{i,j}| \Delta_{i,j}^c - G_{i,j} |\Delta_{i,j}^s|$ and $|B_{i,j}| \Delta_{i,j}^c + G_{i,j} |\Delta_{i,j}^s|$ should be non-negative. That is $|\frac{G_{i,j}}{B{i,j}}| \leq |\cot(\theta_{i,j})|$. 
\end{proof}
\vspace{5mm}
To interpret the above result, consider an example where $\frac{G_{i,j}}{B{i,j}} \leq 1$. The above result implies that system \eqref{eq:netcT} is cooperative if $\theta_{i,j}(t) \in [-\frac{\pi}{4},\frac{\pi}{4}]$.

\begin{remark}
The result in Proposition \ref{pr5} restricts the variation of voltage angles based on $\frac{G_{i,j}}{B_{i,j}}$ ratio of power lines. One potential solution to relax this restriction is to consider the combination of both active and reactive power, e.g. $P_i+Q_i$, as the control input. Studying this possible extension is among our future avenues.
\end{remark}
\section{Analysis: The case of decoupled power flow}\label{sec:dec}
In this section, we present stability results for system \eqref{eq:netc} assuming a decoupled power flow model such that ${\theta}_{i,j}=0$. The latter is the assumption behind the design of the controller in \eqref{eq:qdr2} \cite{jw16}. We also, assume that $\dot k_i=0$, ${\dot V}^\ast_i=0$. Without loss of generality, we take $\diag(\tau)$ as an identity matrix. The network model in this case is
\begin{equation}\label{eq:netc2}
\dot V= \diag(V)(\Psi^\ell\;V+b),
\end{equation}
where the interaction matrix $\Psi^\ell$ is as follows
\begin{equation}\begin{aligned}\label{eq:net4}
\Psi^\ell= &\begin{bmatrix} -(|B_1|+k_1)& |B_{1,2}| & \ldots & |B_{1,n}|  \\ |B_{1,2}| &-(|B_2|+k_2) & \ldots & |B_{2,n}| \\ \vdots & \vdots & \cdots & \vdots \\|B_{1,n}| & |B_{2,n}|& \ldots & -(|B_n|+k_n)\end{bmatrix}.
\end{aligned}\end{equation}
\begin{proposition}\label{pr6}
If $\forall i: k_i >0$, then matrix $\Psi^\ell$ in \eqref{eq:net4} is negative definite.
\end{proposition}
\begin{proof}
The proof follows a similar trend as the proof of Proposition \ref{pr2}.
\end{proof}
\begin{corollary}
System \eqref{eq:netc2} is a stably dissipative Lotka-Volterra system.
\end{corollary}
\begin{proof}
If $k_i>0$, $\Psi^\ell <0$, hence the system is dissipative. Moreover, since $-(|B_i|+k_i) < 0$, based on Theorem 2.1 of \cite{zhao2010classification}, system \eqref{eq:netc2} is stably dissipative.\\
\end{proof}
Now, let us investigate conditions under which the system is cooperative and provide a sufficient
condition for existence of an equilibrium in ${\rm int}(\R^n_{+})$.
\begin{proposition}\label{pr7}
If  $\forall i: k_i V_i^\ast > 0$, then system \eqref{eq:netc2} is cooperative and there exists an equilibrium point $\bar V$ of system \eqref{eq:netc2} which is unique in ${\rm int}(\R^n_{+})$. In particular, if $B_i^{sh}=0$ and $V^\ast_i=V^\ast$, then $V^\ast$ is the unique equilibrium for \eqref{eq:netc2}.
\end{proposition}
\begin{proof}
Based on Definition \ref{def:co}, system \eqref{eq:netc2} is cooperative if the interaction matrix $\Psi^\ell$ is Metzler (see Definition \ref{def:met}). Since, $|B_{i,j}| \geq 0 $, then $\Psi^\ell$ is Metzler. Further, based on Theorem 6.5.3 of \cite{luenberger1979introduction}, if $\Psi^\ell$ is Metzler and Hurwitz, then $\Psi^{-\ell}$ is Hurwitz and $- \Psi^{-\ell} > 0$. From Proposition \ref{pr6}, $\{\forall i: k_i> 0\}$, $\Psi^\ell$ is Hurwitz. Therefore, the proof is completed if every element of vector $b$ in \eqref{eq:netc2} is positive, that is $k_i V_i^\ast >0$. Considering the specific case where $B_i^{sh}=0$ and $V^\ast_i=V^\ast$, the proof is straightforward since $|B_i|= \sum_{j \in {\cal N}_i} |B_{i,j}|$ holds.   
\end{proof}
\begin{remark}\label{mon}[Monotonicity of system \eqref{eq:netc2}]: 
The conditions of Proposition \ref{pr7} guarantee that system \eqref{eq:netc2} is cooperative, \ie, $\Psi^\ell$ is Metzler (Definition \ref{def:co}). Hence, the flow of system \eqref{eq:netc2} is monotone, that is given two initial conditions $x_0, y_0 \in {\rm int}(\R^n_{+})$, $x_0 \geq y_0$ (element-wise) implies that $x(t,x_0) \geq x(t,y_0)$ for all $t$. Notice that for linear time-invariant systems, a positive system is also cooperative and monotone, however a nonlinear positive system is not necessarily monotone \cite{angeli2003monotone}. 
\end{remark}
Now, we present a Lyapunov-based stability analysis assuming the existence of a positive equilibrium. Compared to \cite{de2017bregman,jw16}, the following result uses a different Lyapunov function which is defined in ${\rm int}(\R^n_{+})$.
\begin{proposition}\label{pr8}
The unique equilibrium point $\bar V$ for system \eqref{eq:netc2} in ${\rm int}(\R^n_{+})$ is asymptotically stable with the domain of attraction equal to
${\rm int}(\R^n_{+})$.
\end{proposition}
\begin{proof}
Assume $\bar V$ is the unique equilibrium of \eqref{eq:netc2} in ${\rm int}(\R^n_{+})$, that is $\Psi^\ell \bar V+b=0$. Take ${\cal V}= \sum _{i} (V_i-{\bar V}_i)- {\bar V}_i (\ln {V_i}-\ln {\bar V}_i)$ as the Lyapunov candidate. The function $\cal V$ defined on $\R^n_{+}$ has the following properties: 
${\cal V}(0) \rightarrow +\infty$, ${\cal V}(+\infty) \rightarrow +\infty$,
${\cal V} (V) \geq 0$, and ${\cal V}(\bar V)=0$.\\
Let calculate the derivative of $\cal V$ as follows
\be\begin{aligned}
\dot{\cal V}&= \mathbf{1}^T \dot V- {\bar V}^T \diag^{-1}(V) \dot{V}\\
&=\mathbf{1}^T \diag(V) \diag^{-1}(V) \dot V- {\bar V}^T \diag^{-1}(V) \dot{V}\\
&= (V -\bar V)^T \diag^{-1}(V) \dot V\\
&= (V -\bar V)^T (\Psi^\ell V+b).
\end{aligned}\ee
Recall that $\Psi^\ell < 0$. Also, from the definition of the equilibrium, we have $\Psi^\ell \bar V =-b$. Hence, we obtain $$\dot{\cal V}=(V-\bar V)^T \Psi^\ell (V-\bar V) \leq 0$$ which ends the proof.
\end{proof}
\section{Simulation results}\label{sec:sim}
This section presents simulation results for a network of five nodes as in Figure \ref{fig:net}. The initial conditions for the nodal voltages are $V(0)=(1.8,1.6,1.4,1.2,1)^T$. We set the lines' suceptances and conductances as $B_{1,2}= -1.5, B_{1,3}= -1, B_{2,3}= -0.7, B_{3,4}= -1.8, B_{4,5}= -1.2$ and $G_{i,j}= 0.5 |B_{i,j}|$. Shunt susceptances are set to zero. 
\begin{figure}[h]
\centering
\includegraphics[scale=0.52]{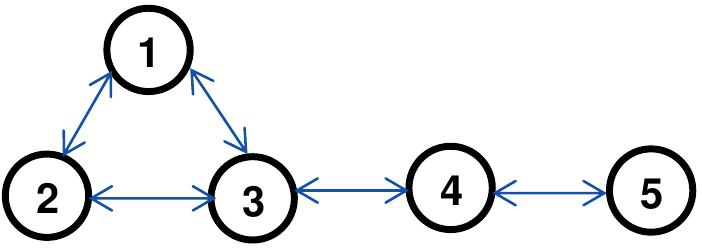}
\caption{Network topology.}\label{fig:net}
\end{figure} 
Figure \ref{fig:osc} shows the result of Proposition \ref{pr4} with $\theta_{i,j}= \theta_{i,j}(0)+ \frac{\pi}{10} \sin(120 t)$ where $\theta(0)=(\frac{\pi}{20},\frac{\pi}{25},\frac{\pi}{30},\frac{\pi}{35},\frac{\pi}{40})^T$. The reference, $V^\ast_i(t)$, is equal to $2+ 0.2 \sin(t)$ for nodes $1,3,5$ and equal to $2+ 0.2 \cos(t)$ for nodes $2,4$. As shown, the time-varying system is bounded.
\begin{figure}[h]
\centering
\includegraphics[scale=0.35]{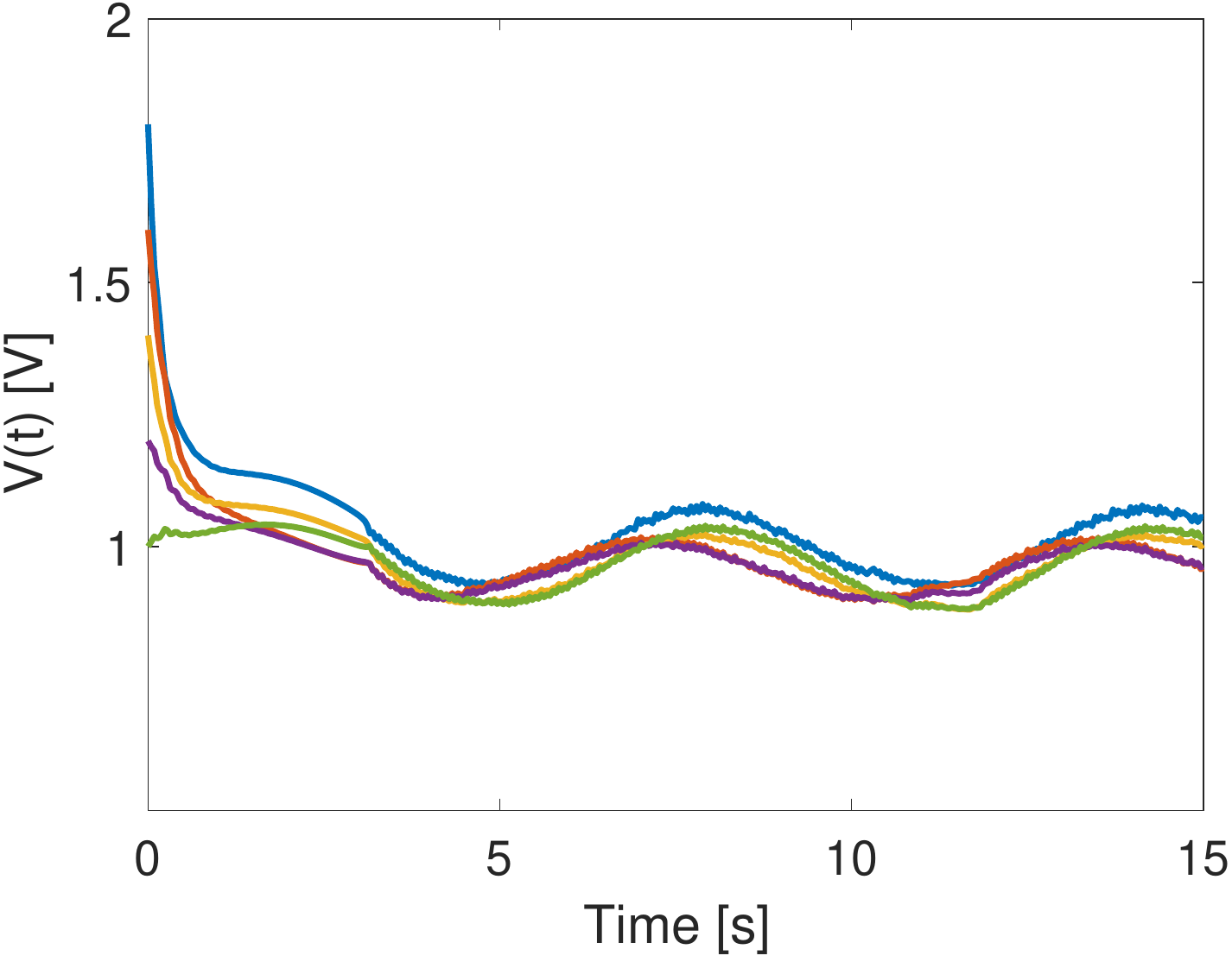}
\caption{The result of Proposition \ref{pr4} with time-varying relative angles and references. As shown the system is bounded.}\label{fig:osc}
\end{figure}
To verify the results of Proposition \ref{pr5}, we replace $k_i, V^\ast_i$ with constant values such that $k_i=5$ and $V^\ast_i=2$. Figure \ref{fig:metz1} shows the evolution of nodal voltages with the controller \eqref{eq:qdr} with constant droop gains and references. As shown, the trajectories are bounded and converging to a ball in the vicinity of the desired equilibrium. 
\begin{figure}[h]
\centering
\includegraphics[scale=0.37]{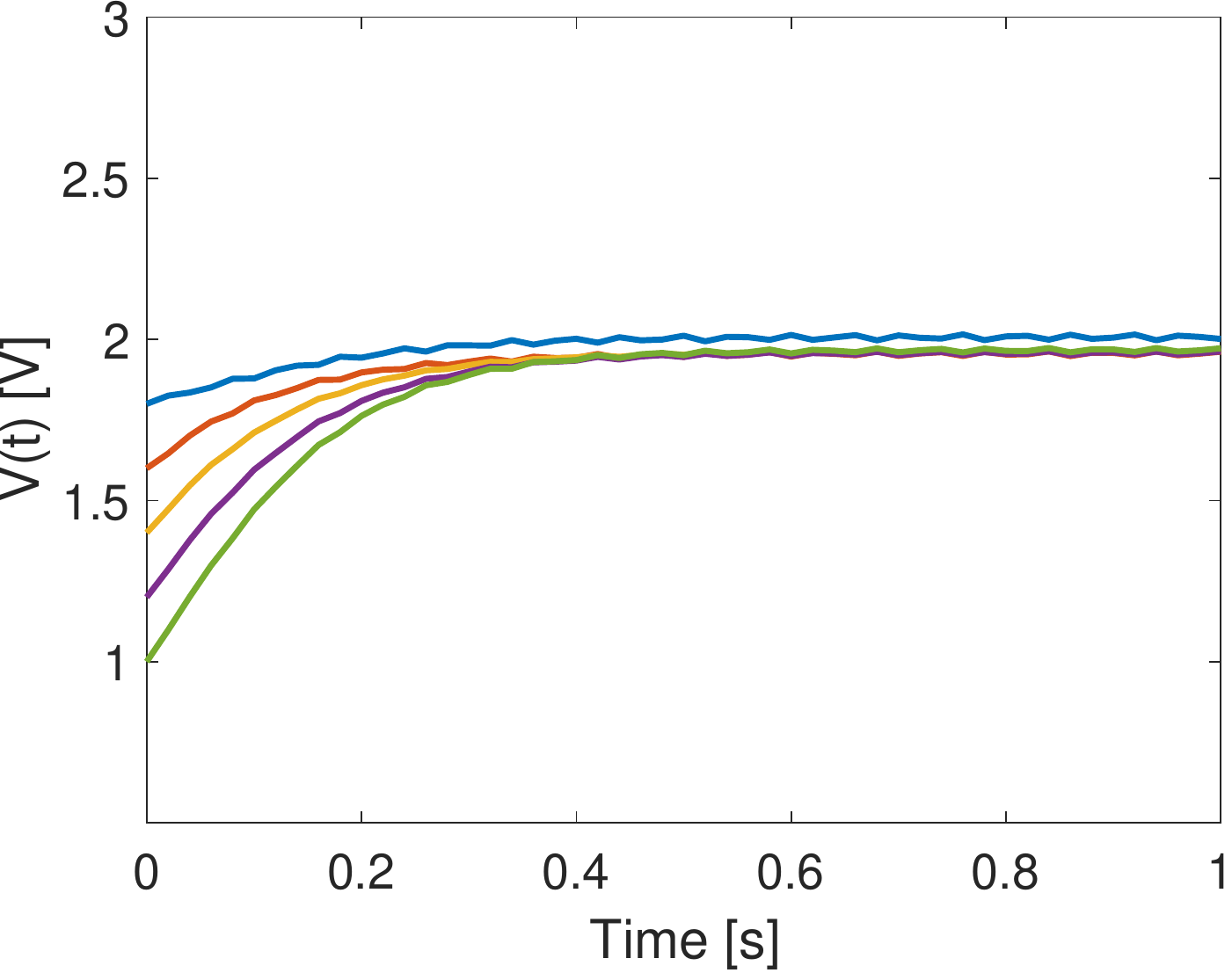}
\caption{Nodal voltages with controllers \eqref{eq:qdr}.}\label{fig:metz1}
\end{figure}
Figure \ref{fig:metz2} shows the result of the case where the controller in \eqref{eq:qdr2} is used (Proposition \ref{pr7}). The line conductances are set to zero and $\theta_{i,j}=0$. Similar to the previous case, $k_i=5$, and $V^\ast_i=2$. The interaction matrix $\Psi^\ell$ is Metzler and Hurwitz. Here, the voltages converge to the reference $V^\ast_i=2$. Also, the results are shown for two sets of initial conditions $V_1(0)=(1.8,1.6,1.4,1.2,1)^T$ and $V_2(0)=(2.8,2.6,2.4,2.2,2)^T$ to show that the system is cooperative and monotone (see Remark \ref{mon}).
\begin{figure}[h]
\centering
\includegraphics[scale=0.37]{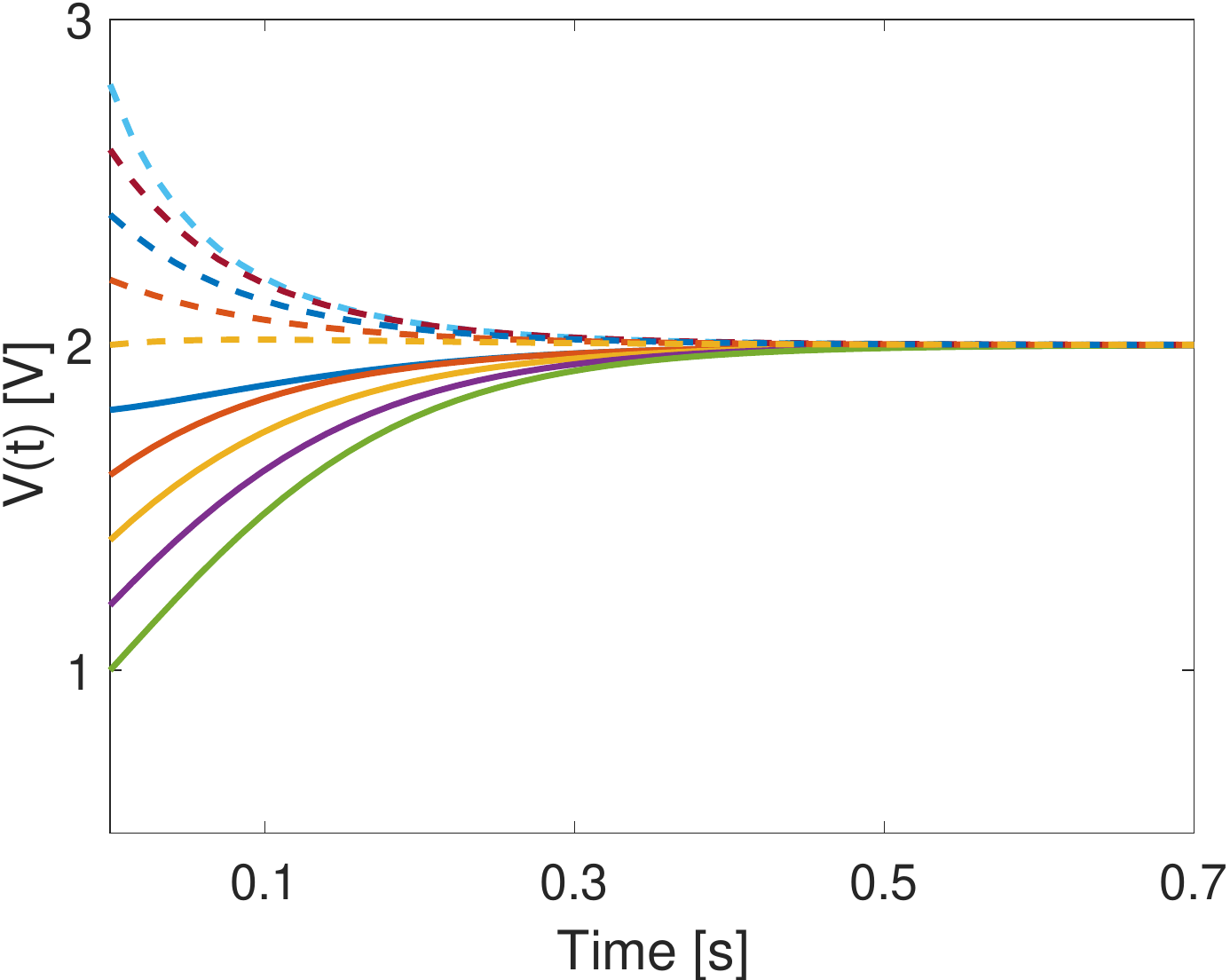}
\caption{Nodal voltages with controllers \eqref{eq:qdr2}. Matrix $\Psi^\ell$ is Metzler and Hurwitz, and the network is cooperative.}\label{fig:metz2}
\end{figure}
\section{Conclusions}\label{sec:con}
This paper has studied the stability of a power network whose nodal voltages are controlled by quadratic droop controllers with injection of AC reactive power. We have shown that the nonlinear voltage dynamics is a positive system in the form of a Lotka-Volterra system and studied its stability. For the lossless network with zero relative angles,
the existence and stability of the unique equilibrium have been proved. For the lossy time-varying network,
we have proved an ultimate uniform boundedness result. Future research avenues include characterizing the ultimate bound for the time-varying system and considering a network with heterogeneous controllers.  
\bibliographystyle{plain}
\bibliography{biblio}

\begin{thebibliography}{10}
\providecommand{\url}[1]{#1}
\csname url@samestyle\endcsname
\providecommand{\newblock}{\relax}
\providecommand{\bibinfo}[2]{#2}
\providecommand{\BIBentrySTDinterwordspacing}{\spaceskip=0pt\relax}
\providecommand{\BIBentryALTinterwordstretchfactor}{4}
\providecommand{\BIBentryALTinterwordspacing}{\spaceskip=\fontdimen2\font plus
\BIBentryALTinterwordstretchfactor\fontdimen3\font minus
  \fontdimen4\font\relax}
\providecommand{\BIBforeignlanguage}[2]{{%
\expandafter\ifx\csname l@#1\endcsname\relax
\typeout{** WARNING: IEEEtran.bst: No hyphenation pattern has been}%
\typeout{** loaded for the language `#1'. Using the pattern for}%
\typeout{** the default language instead.}%
\else
\language=\csname l@#1\endcsname
\fi
#2}}
\providecommand{\BIBdecl}{\relax}
\BIBdecl

\bibitem{andreasson2017performance}
M.~Andreasson, E.~Tegling, H.~Sandberg, and K.~Johansson, ``Performance and
  scalability of voltage controllers in multi-terminal {HVDC} networks,'' in
  \emph{American Control Conference}, 2017, pp. 3029--3034.

\bibitem{de2017bregman}
C.~De~Persis and N.~Monshizadeh, ``Bregman storage functions for microgrid
  control,'' \emph{IEEE Transactions on Automatic Control}, 2017.

\bibitem{matincdc2016}
M.~Jafarian, J.~Scherpen, and M.~Aiello, ``A price-based approach for voltage
  regulation and power loss minimization in power distribution networks,'' in
  \emph{55th Conference on Decision and Control}, 2016, pp. 680--685.

\bibitem{schiffer2014conditions}
J.~Schiffer, R.~Ortega, A.~Astolfi, J.~Raisch, and T.~Sezi, ``Conditions for
  stability of droop-controlled inverter-based microgrids,'' \emph{Automatica},
  vol.~50, no.~10, pp. 2457--2469, 2014.

\bibitem{vasquez2009adaptive}
J.~Vasquez, J.~Guerrero, A.~Luna, P.~Rodr{\'\i}guez, and R.~Teodorescu,
  ``Adaptive droop control applied to voltage-source inverters operating in
  grid-connected and islanded modes,'' \emph{IEEE Transactions on Industrial
  Electronics}, vol.~56, no.~10, pp. 4088--4096, 2009.

\bibitem{kundur1994power}
P.~Kundur, N.~Balu, and M.~Lauby, \emph{Power system stability and
  control}.\hskip 1em plus 0.5em minus 0.4em\relax McGraw-hill New York, 1994,
  vol.~7.

\bibitem{jw16}
J.~W. Simpson-Porco, F.~D{\"o}rfler, and F.~Bullo, ``Voltage stabilization in
  microgrids via quadratic droop control,'' \emph{{IEEE} Transactions on
  Automatic Control}, vol.~62, no.~3, pp. 1239--1253, 2016.

\bibitem{teixeira2015voltage}
A.~Teixeira, K.~Paridari, H.~Sandberg, and K.~Johansson, ``Voltage control for
  interconnected microgrids under adversarial actions,'' in \emph{20th IEEE
  Conference on Emerging Technologies \& Factory Automation}, 2015, pp. 1--8.

\bibitem{angeli2003monotone}
D.~Angeli and E.~Sontag, ``Monotone control systems,'' \emph{IEEE Transactions
  on automatic control}, vol.~48, no.~10, pp. 1684--1698, 2003.

\bibitem{de2001stability}
P.~De~Leenheer and D.~Aeyels, ``Stability properties of equilibria of classes
  of cooperative systems,'' \emph{IEEE Transactions on Automatic Control},
  vol.~46, no.~12, pp. 1996--2001, 2001.

\bibitem{peuteman2000boundedness}
J.~Peuteman, D.~Aeyels, and R.~Sepulchre, ``Boundedness properties for
  time-varying nonlinear systems,'' \emph{SIAM Journal on Control and
  Optimization}, vol.~39, no.~5, pp. 1408--1422, 2000.

\bibitem{cui2001permanence}
J.~Cui and L.~Chen, ``Permanence and extinction in logistic and
  {L}otka-{V}olterra systems with diffusion,'' \emph{Journal of Mathematical
  Analysis and Applications}, vol. 258, no.~2, pp. 512--535, 2001.

\bibitem{zhao2010classification}
X.~Zhao and J.~Luo, ``Classification and dynamics of stably dissipative
  {L}otka-{V}olterra systems,'' \emph{International Journal of Non-Linear
  Mechanics}, vol.~45, no.~6, pp. 603--607, 2010.

\bibitem{varga2010gervsgorin}
R.~Varga, \emph{Ger{\v{s}}gorin and his circles}.\hskip 1em plus 0.5em minus
  0.4em\relax Springer Science \& Business Media, 2010, vol.~36.

\bibitem{bacciotti1999stability}
A.~Bacciotti and F.~Ceragioli, ``Stability and stabilization of discontinuous
  systems and nonsmooth lyapunov functions,'' \emph{ESAIM: Control,
  Optimisation and Calculus of Variations}, vol.~4, pp. 361--376, 1999.

\bibitem{jafarian2015formation}
M.~Jafarian, E.~Vos, C.~De~Persis, A.~van~der Schaft, and J.~Scherpen,
  ``Formation control of a multi-agent system subject to coulomb friction,''
  \emph{Automatica}, vol.~61, pp. 253--262, 2015.

\bibitem{luenberger1979introduction}
D.~Luenberger, ``Introduction to dynamic systems: theory, models, and
  applications,'' 1979.

\end{thebibliography}
\end{document}